\documentclass[a4paper,11pt]{amsart}
\usepackage{a4wide}
\usepackage[utf8]{inputenc}
\usepackage{amsthm}
\usepackage{amsmath, amssymb}
\usepackage{mathrsfs}
\usepackage{microtype}
\usepackage{graphicx,wrapfig,lipsum}
\usepackage{latexsym}
\usepackage{pstricks}
\usepackage{amsthm}
\usepackage{etoolbox}
\usepackage{caption}
\usepackage{subcaption}
\usepackage{resizegather}

\usepackage{parskip}
\usepackage[justification=centering]{caption}
\allowdisplaybreaks

\newtheorem{theorem}{Theorem}[section]

\newtheorem{lemma}[theorem]{Lemma}
\newtheorem{proposition}[theorem]{Proposition}
\newtheorem{definition}[theorem]{Definition}

\newtheoremstyle{myrem}%name
 {7pt}%Space above
 {7pt}%Space below
 {}%Body font
 { }%Indent amount
 {\bf}% Theorem head font
 {.}%Punctuation after theorem head
 { }%Space after theorem head 2
 {}%Theorem head spec (can be left empty, meaning ‘normal’)

 \theoremstyle{myrem}
 \newtheorem*{remark}{Remark}
 
 \newtheoremstyle{myrem}%name
 {7pt}%Space above
 {7pt}%Space below
 {}%Body font
 { }%Indent amount
 {\bf}% Theorem head font
 {.}%Punctuation after theorem head
 { }%Space after theorem head 2
 {}%Theorem head spec (can be left empty, meaning ‘normal’)

 \theoremstyle{myrem}
 \newtheorem*{example}{Example}
\title{A Pfaffian formula for the monomer-dimer model on surface graphs}
\author{Anh Minh Pham}
\address{Université de Genève, Section de mathématiques, 2 rue du Lièvre, 1211 Genève, Switzerland}
\email{AnhMinh.Pham@unige.ch}
\begin{document}
\begin{abstract}
We consider the monomer-dimer model on weighted graphs embedded in surfaces with boundary, with the restriction that only monomers located on the boundary are allowed. We give a Pfaffian formula for the corresponding partition function, which generalises the one obtained by Giuliani, Jauslin and Lieb for graphs embedded in the disk \cite{Lieb16}. Our proof is based on an extension of a bijective method mentioned in \cite{Lieb16}, together with the Pfaffian formula for the dimer partition function of Cimasoni-Reshetikhin \cite{CimRes07}.
\end{abstract}
\maketitle
\section{Introduction}
The monomer-dimer model is one of the important models in both statistical physics and theoretical computer science. As a classical lattice model, it was first introduced to describe the absorption of a mixture of unequal-size molecules on crystal surfaces \cite{Fowler37}. A brief history of the study of this model in statistical mechanics, and some fundamental results can be found in \cite{Hielmann-Lieb70,Heilmann-Lieb72}. Moreover, the impact and the interest of the monomer-dimer model goes beyond physics: indeed the monomer-dimer problem represents a prototypical problem in computational complexity theory \cite{Garey02}.

This model can be defined as follows. Let $G$ be a finite graph with vertex set $V(G)$ and edge set $E(G)$; a $\emph{monomer-dimer covering}$, or shortly, an MD covering of $G$ is a pair $\tau=(\tau_D,\tau_M)\in E(G) \times V(G)$ so that each vertex of $G$ is covered by exactly one element of $\tau_D$ (which we shall call a $\emph{dimer}$) or one element of $\tau_M$ (which we shall call a $\emph{monomer}$). Let us denote by $\mathcal{MD}(G)$ the set of all MD coverings of $G$. If monomers are not allowed to appear, that is, if $\tau_M=\emptyset$, then we simply say that $\tau=\tau_D$ is a $\emph{dimer covering}$ (aka $\emph{perfect matching}$) of $G$. We denote by $\mathcal{D}(G)$ the set of all dimer coverings of $G$. If $G$ is endowed with an edge weight system $x=(x_e)_{e\in E(G)}$ and a vertex weight system $y=(y_v)_{v\in V(G)}$, then the monomer-dimer partition function (or simply the MD partition function) of the weighted graph $(G,x,y)$ is defined by $$Z_{\mathcal{MD}}(G):=Z_{\mathcal{MD}}(G,x,y)=\sum_{\tau\in \mathcal{MD}(G)}\prod_{e\in \tau_D}x_e\prod_{v\in \tau_M}y_v.$$ Note that if we set all the vertex weights $y=(y_v)_{v\in V(G)}$ equal to 0, which means that monomers are not allowed to appear, then we get the dimer partition function of $G$.

A classical problem associated to the monomer-dimer model on a given weighted graph $G$ is to compute its partition function $Z_{\mathcal{MD}}(G)$. Unfortunately, it turns out that this computation, even for planar graphs, is intractable \cite{Jer87}. More precisely, it is proved in \cite{Jer87} that computing the number of matchings (and so computing the number of MD coverings) for planar graphs is ``$\#$P complete''. Nevertheless, there are some particular cases in which we can compute the MD partition function efficiently. Let us now discuss these cases.

If monomers are not allowed to appear, or if they are fixed, computing the monomer-dimer partition function boils down to computing the dimer partition function of some resulting graph. For the latter, one can use a $\emph{Pfaffian formula}$, which expresses the dimer partition function of any surface graph (i.e. graph embedded in a surface) as a linear sum of the Pfaffian of some skew-adjacency matrices associated to the graph. In fact this formula was first achieved by Fisher \cite{Fis61} and Kasteleyn \cite{Kas61} for the square lattice, and then by Kasteleyn for every planar graph \cite{Kas63}. The result was extended later to general surface graphs by Galluccio-Loebl \cite{Gal99} and Tesler \cite{Tes2000} independently, and then by Cimasoni-Reshetikhin \cite{CimRes07,CimRes08, Cim09}.

Furthermore, if monomers are allowed to appear, and if one restricts their locations, similar Pfaffian formulas for MD partition functions can be obtained in some cases \cite{Tem73, Wu06, Priezzhev08, WuTzeng11, Lieb16}. To the best of our knowledge, the most general Pfaffian formula for the MD partition function was obtained only recently in the paper \cite{Lieb16} of Giuliani, Jauslin and Lieb where the authors consider graphs embedded in the disk with monomers restricted on the boundary of this disk. There are two methods to prove their Pfaffian formula. One method is to define a suitable matrix associated to a well-chosen orientation on the edges of the graph together with a good labelling of its vertices so that the Pfaffian of this matrix is equal to the MD partition function (see \cite{Lieb16} for more details). Another method is based on a one-to-two mapping from the set of MD coverings of the original graph to the set of dimer coverings of an auxiliary planar graph, and the application of the Pfaffian formula for the dimer partition function to this auxiliary graph (see \cite[Appendix E]{Lieb16}). 

The aim of the present article is to generalise the Pfaffian formula obtained by Giuliani, Jauslin and Lieb to graphs embedded in surfaces of arbitrary topology with any number of boundary components, where monomers are still restricted on the boundary of the surface. Roughly speaking, if $G$ is embedded in an orientable surface of genus $g$ with $b$ boundary components, the MD partition function of $G$ (with monomers located only on the boundary) is given by $$Z_{\mathcal{MD}}(G)=\frac{1}{2^g}\bigg|\sum_{K}\pm \text{Pf} (M^K(G)) \bigg|.$$ In this formula, the sum is over $2^{2g+b-1}$ well-chosen orientations on $G$, and $M^K(G)$ is a modified adjacency matrix of $G$ with respect to the orientation $K$ (see Definition \ref{def: modified} below). The more precise statement of our formula can be found in Theorem \ref{theo: main}. As a consequence, if monomers are not allowed to appear, we get back the general Pfaffian formula for the dimer partition function obtained in \cite{Cim09}. Also, if $g=0$ and $b=1$ we get back the Pfaffian formula for the MD partition function given in \cite{Lieb16}.

Before concluding this introduction, let us give the idea of the proof of our formula. In short, by extending the bijective method described in \cite[Appendix E]{Lieb16} to graphs embedded in surfaces with boundary, we can define an auxiliary surface graph $G_{\beta}$ for each $\beta=(\beta_1,\dots,\beta_b)\in \mathbb{Z}_2^b$  with $\sum_{k=1}^b\beta_k$ even with the following property:  the dimer partition function of $G_{\beta}$ can be related to a partial MD partition function of $G$ with the parity of the number of monomers on each boundary component given by $\beta$. The former quantity then can be computed by using the Pfaffian formula for the dimer partition function given in \cite{Cim09}. It follows that the MD partition function of $G$ can be given by a linear sum of the Pfaffian of skew-adjacency matrices associated to $G_{\beta}$. Furthermore, one can show that the latter can be related to the Pfaffian of modified adjacency matrices of $G$ mentioned above, and hence we end up with the formula as stated. 

In conclusion, it should be mentioned that using the same method as described above together with the Pfaffian formula for the dimer partition function of graphs embedded in (possibly) non-orientable surfaces \cite{Cim09}, one can obtain a similar formula for the MD partition function of graphs embedded in (possibly) non-orientable surfaces with boundary. The proof of this formula is almost the same as in the orientable case, and therefore we only consider the latter in the present work. Last but not least, for the dimer partition function of surface graphs, the number of terms in the Pfaffian formula is precisely the order of the first homology group over $\mathbb{Z}_2$ of the surface. By our formula, this is also the case when we consider the monomer-dimer model.

The article is organised as follows. In Section \ref{sec: statement} we define modified adjacency matrices associated to graphs embedded in surfaces with boundary. We also describe particular orientations on edges of graphs together with a specific labelling of vertices. Finally we state our general Pfaffian formula for the monomer-dimer partition function (Theorem \ref{theo: main}). Section \ref{sec: proof} is devoted to review the dimer model on surface graphs (Section \ref{subsec: dimer}), and then to prove our formula (Section \ref{subsec: main proof}).
\subsection*{Acknowledgments} This work was supported by a grant of the Swiss National Science Foundation. The author would also like to thank his advisor David Cimasoni for helpful comments and discussions.
\section{Definitions and statement of the main result}\label{sec: statement}
In this section we will fix the setting for the rest of the paper, and introduce some basic definitions and notations that are necessary to state our result. We then give our Pfaffian formula for the MD partition function together with an example at the end of the section.
\subsection{Setting} Throughout this article, we shall assume that $G$ is a finite connected graph endowed with an edge weight system $x=(x_e)_{e\in E(G)}$ and a vertex weight system $y=(y_v)_{v\in V(G)}$. Note that $G$ is allowed to have multiple edges, while loops can always be removed since they do not play any role in our model. We also assume that $G$ is embedded in an orientable surface $\Sigma$ of genus $g$ with boundary $\partial \Sigma$ consisting of $b\geq 1$ components. Let us define the boundary $\partial G$ of $G$ as the subgraph of $G$ consisting of vertices and edges that can be connected to $\partial \Sigma$ by a path without crossing any other edge of $G$ (cf. \cite{Lieb16}).

\begin{figure}[b]
\centering
  \includegraphics[height=115pt]{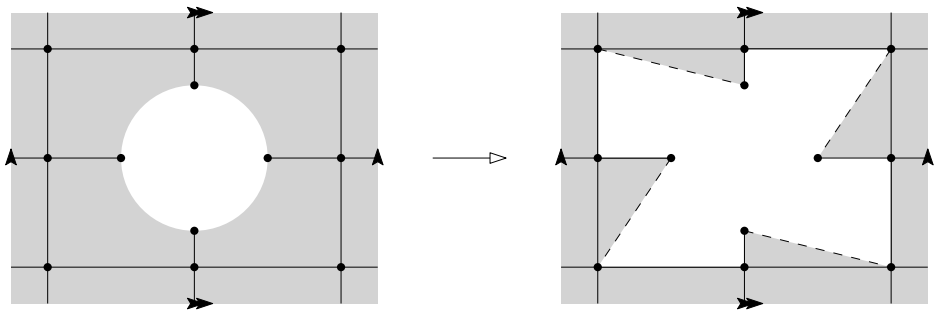}
  \caption{Add suitable (dashed) edges and deform $\partial \Sigma$ to get a boundary circuit.} \label{fig: boundarycircuit}
\end{figure}
Following \cite{Lieb16}, one can make the following assumptions without any loss of generality. First of all, by adding suitable edges of weights 0 and by deforming the boundary of $\Sigma$ if needed, we can assume that $\partial G$ consists of $b$ circuits $B_1,\dots,B_b$ which coincide with the $b$ components $C_1,\dots,C_b$ of $\partial\Sigma$ respectively (see Figure \ref{fig: boundarycircuit} when $g=1$, $b=1$). This means that one can travel along each boundary component of $G$ (as well as that of $\Sigma$) visiting each vertex in this component exactly once. We refer the reader to \cite[Section 4.1]{Lieb16} for an explicit algorithm to construct a boundary circuit for any planar graph, which is also valid for any graph embedded in surfaces with boundary. Note that by this first assumption, we can identify $\partial G$ with $\partial\Sigma$, and hence the MD partition function with monomers restricted to $\partial\Sigma$ coincides with the ``boundary MD partition function'' defined in \cite{Lieb16} for $g=0$ and $b=1$. For this reason we shall also call the former one (which we want to compute) the boundary MD partition function, and still denote it by $Z_{\mathcal{MD}}(G)$. Let us now continue with the second assumption. Denoting by $N_k$ the number of vertices of $G$ on $B_k$ for $1\leq k\leq b$, one can assume that $N_k$ is even for every $k$. Indeed, if $N_k$ is odd for some $k$, let $G'$ be the graph obtained by transforming $G$ as in Figure \ref{fig: evenboundaryvertex}: add an edge of weight 1 with one endpoint on $B_k$ and the other one in the interior of $\Sigma$. This added edge must be occupied by every MD covering of $G'$ whose monomers are restricted to $\partial G'$. Then one can verify easily that $Z_{\mathcal{MD}}(G)=Z_{\mathcal{MD}}(G')$, and that $G'$ has an even number of vertices on its newly-created boundary circuit $B_k'$. Finally, we can assume that $|V(G)|$ is even: otherwise, let $G''$ be the surface graph obtained by adding a vertex of weight 1 to an arbitrary face of $G$, and by connecting this vertex to all the vertices on the boundary of this face by edges of weights 0; then it is clear that $Z_{\mathcal{MD}}(G)=Z_{\mathcal{MD}}(G'')$, and that $|V(G'')|$ is even. To sum up, from now on $G$ will be assumed to have an even total number of vertices, as well as an even number of vertices on each of its boundary circuit.

For further purpose, let us denote by $\overline{\Sigma}$ the closed surface of genus $g$ obtained from $\Sigma$ by gluing a topological 2-disc $D_k$ along each boundary component $C_k$ for every $1\leq k\leq b$.
\begin{figure}[t]
\centering
  \includegraphics[height=130pt]{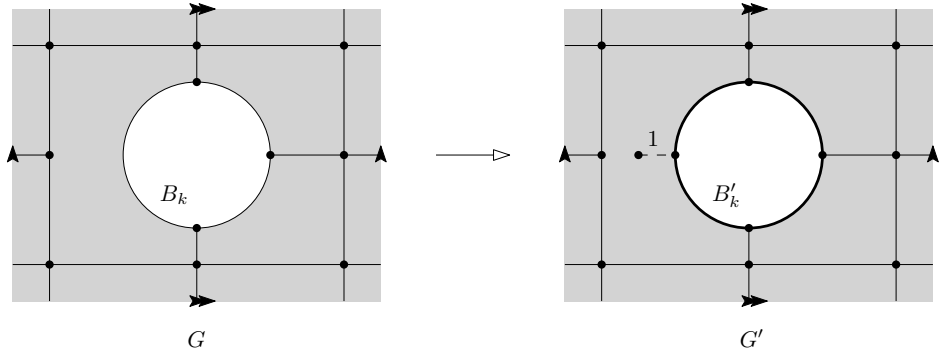}
  \caption{Add an edge (in dashed line) to the boundary circuit $B_k$ of $G$ to get the graph $G'$ with the new boundary circuit $B_k'$ (in heavily bolded arcs).} \label{fig: evenboundaryvertex}
\end{figure}
\subsection{Modified adjacency matrices}
We now define modified adjacency matrices of the graph $G$ with respect to some labelling of its vertices and some orientation $K$ on its edges. These matrices are actually inspired by the work of Giuliani, Jauslin and Lieb (cf. \cite[Theorem 1.1]{Lieb16}), however we slightly adjust their definition, and extend it to the general case when $G$ has more than one boundary circuit. More precisely, our matrices, which are parametrised by elements of $\mathbb{Z}_2^b$, are defined as follows.
\begin{definition} \label{def: modified}
 Let $(G,x,y)$ be a weighted graph and $K$ an orientation on its edges. Given $\beta=(\beta_1,\dots,\beta_b)\in \mathbb{Z}_2^b$ with $\sum_{k=1}^b \beta_k$ even, add an isolated vertex $v_k$ of weight 1 for each $k$ such that $\beta_k=1$. Labelling the vertices of $G$ together with all such $v_k$'s by a number set $I\subset \mathbb{N}$, we define the $\emph{modified adjacency matrix}$ $M^K_{\beta}(G)=(m^{\beta}_{ij})_{i,j\in I}$ of $G$ as the skew-symmetric matrix whose entries are given by $$m^{\beta}_{ij}= \sum_{e}\epsilon^K_{ij}(e)x_e+(-1)^{i+j}d_{ij}y_iy_j \quad \text{for $i<j$},$$ where the sum is taken over all the edge $e$ between two vertices $i,j$ while $\epsilon^K_{ij}(e)$ is equal to 1 if $e$ is oriented by $K$ from $i$ to $j$, and equal to -1 otherwise. Moreover we define  $$d_{ij}= \left\{  \begin{array}{ll}
        1& \text{if $i,j$ belong to a same boundary component};\\
        0& \text{if not},
\end{array} \right.$$ where the ``virtual vertex'' $v_k$ is always considered to be on the boundary component $B_k$.
 \end{definition}
 \begin{figure}[t]
\centering
  \includegraphics[height=120pt]{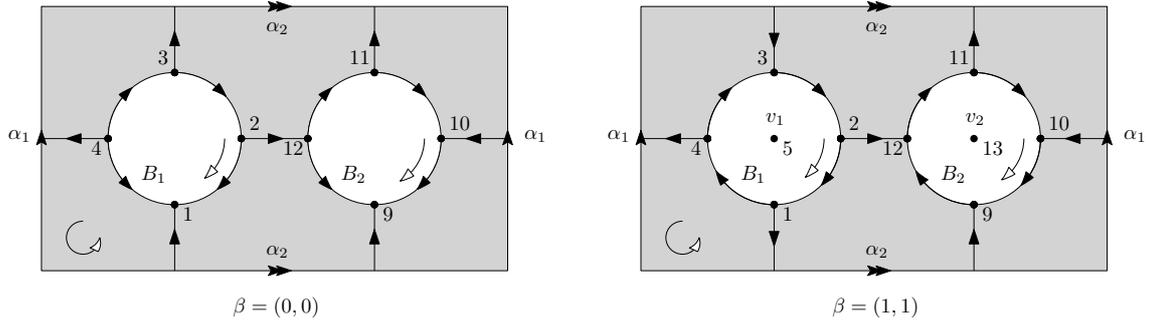}
  \caption{Labellings and orientations of interest for $b=2$. $\gamma_1=(1,4,3)$ and $\gamma_2=(4,3,2,12,11,10)$ are oriented cycles having $\alpha_1$ and $\alpha_2$ on their immediate left.} \label{fig: orientationandlabelling}
\end{figure}
Note that by definition, $M^K_{\beta}(G)$ does not only depend on the parameter $\beta$ but also on the orientation $K$ and on the way we label the vertices of $G$ (and the $v_k$'s). We now describe a specific labelling and an orientation of interest.
 
Let us begin with a specific labelling of the vertices of $G$. Recall that $G$ is embedded in the orientable surface $\Sigma$ with boundary $\partial\Sigma$ consisting of $b$ circuits $B_k$'s. Then endowing $\Sigma$ with an orientation (pictured counterclockwise) induces a natural orientation (pictured clockwise) on each $B_k$ (cf. Figure \ref{fig: orientationandlabelling}). For each parameter $\beta$, we label the vertices on each $B_k$ (together with $v_k$ if $\beta_k=1$) by increasing consecutive numbers so that we see these vertices in that order when we travel along $B_k$ with its inverse orientation (and then $v_k$ if $\beta_k=1$). This labelling is illustrated in Figure \ref{fig: orientationandlabelling}. Note that when travelling along $B_k$, the starting vertex is not important. The vertices of $G$ that do not belong to any $B_k$ then can be labelled in an arbitrary way.
 
We now continue by describing a specific orientation $K$ of interest. Recall that the orientable surface $\Sigma$ is endowed with a pictured-counterclockwise orientation which also induces an orientation on each face $f$ of $G$, as well as its boundary $\partial f$. Firstly, following \cite{Kas61, Kas63, CimRes07} we require $K$ to be $\emph{Kasteleyn}$, that is, for each face $f$ of $G$, the number $n^K(\partial f)$ of edges on its boundary where the orientation of $\partial f$ is different from $K$ is odd. Secondly, for each $\beta=(\beta_1,\dots,\beta_b)\in \mathbb{Z}_2^b$ with $\sum_{k=1}^b \beta_k$ even, we want $K$ to be specific on $B_k$ for $1\leq k\leq b$ as follows: if $\beta_k=0$, then $K$ is from big-labelled vertices to small-labelled vertices; if $\beta_k=1$ then $K$ coincides with the natural orientation induced on $B_k$ (see Figure \ref{fig: orientationandlabelling}). Let us denote by $\mathcal{O}_{\beta}$ the set of all orientations which are Kasteleyn and satisfy the above condition. For  $\beta=(\beta_1,\dots,\beta_b)\in \mathbb{Z}_2^b$ with $\sum_{k=1}^b\beta_k$ even, the set $\mathcal{O}_{\beta}$ (which will be shown to be nonempty by Lemma \ref{lem: exist}) consists of the orientations we are interested in.
\subsection{Main result}
To state our Pfaffian formula, we need to choose a specific orientation in each $\mathcal{O}_{\beta}$, that can be described as follows. Let $\{\alpha_i\}_{1\leq i\leq 2g}$ be a set of closed curves that are transverse to $G$ and avoid $\partial G$ so that their homology classes $\{[\alpha_i]\}_{1\leq i\leq 2g}$ form a basis of $H_1(\overline{\Sigma};\mathbb{Z}_2)$, the first homology group over $\mathbb{Z}_2$ of $\overline{\Sigma}$. For each $1\leq i\leq 2g$, let $\gamma_i$ be an oriented cycle of $G$ having $\alpha_i$ on its immediate left, and choose $K_{\beta}\in \mathcal{O}_{\beta}$ so that $n^{K_{\beta}}(\gamma_i)$ is odd for every $1\leq i\leq 2g$ (see Figure \ref{fig: orientationandlabelling} for example).

However, in general it is not obvious that there exists $K_{\beta}$ as required for each $\beta=(\beta_1,\dots,\beta_b)\in \mathbb{Z}_2^b$ with $\sum_{k=1}^b\beta_k$ even. Therefore, the proof of this fact will be postponed until the next section (see Lemma \ref{lem: exist}). With this specific $K_{\beta}$ and with each $\epsilon=(\epsilon_1,\dots,\epsilon_{2g})\in \mathbb{Z}_2^{2g}$, let us denote by $K_{\beta}^{\epsilon}$ the orientation obtained by inverting $K_{\beta}$ on every edge $e$ each time $e$ crosses $\alpha_i$ with $\epsilon_i=1$. Then one can state our formula as follows.
\begin{theorem} \label{theo: main}
Let $(G,x,y)$ be a weighted graph embedded in an orientable surface $\Sigma$ of genus $g$ with $b$ boundary components. Then the boundary monomer-dimer partition function of $G$ is given by $$Z_{\mathcal{MD}}(G)=\frac{1}{2^g} \sum_{\substack{\beta\in \mathbb{Z}_2^b\\ \sum_{k}\beta_k\, \text{even}}}\bigg| \sum_{\epsilon\in\mathbb{Z}_2^{2g}}(-1)^{\sum_{i<j}\epsilon_i\epsilon_j\alpha_i\cdot\alpha_j}\text{Pf}(M^{K_{\beta}^{\epsilon}}_{\beta}(G))\bigg|.$$ In this formula the first sum is over all $\beta=(\beta_1,\dots,\beta_b)\in \mathbb{Z}_2^b$ such that $\sum_{k=1}^b\beta_k$ is even, while in the second sum, $\alpha_i\cdot\alpha_j$ denotes the intersection number of $\alpha_i$ and $\alpha_j$.
\end{theorem}
Note that by our formula, the number of Pfaffians is equal to $2^{2g+b-1}$, which is exactly the order of $H_1(\Sigma;\mathbb{Z}_2)$. In particular, when $g=0$ and $b=1$, the MD partition function is given by a single Pfaffian: this is precisely the formula stated in \cite{Lieb16}.
\begin{remark} There exists another version of Theorem \ref{theo: main} which can be stated as follows. Let us first denote by $\mathcal{S}$ the set of all $\beta=(\beta_1,\dots,\beta_b)\in \mathbb{Z}_2^b$ with $\sum_{k=1}^b\beta_k$ even so that there exists a matching $D_{\beta}$ of $G$ that covers an even (resp. odd) number of vertices on $B_k$ with $\beta_k=0$ (resp. $\beta_k=1$) for every $1\leq k \leq b$. By \cite[Theorem 3.5]{Cim09} (see also \cite{CimRes08}), one can prove that for each $\beta\in \mathcal{S}$ and for each matching $D_{\beta}$ as above, there exists a natural bijection $\psi:=\psi_{D_{\beta}}$ between the set $\mathcal{Q}(\overline{\Sigma})$ of quadratic forms on $\overline{\Sigma}$ and the set $\{[K]: K\in \mathcal{O}_{\beta}\}$ of all the equivalence classes of orientations in $\mathcal{O}_{\beta}$ (we refer the reader to \cite{CimRes08,Cim09} for more details). Denoting by $\text{Arf}(q)$ the Arf invariant of the quadratic form $q$, then by \cite[Theorem 3.8]{Cim09}, the MD partition function of the weighted graph $(G,x,y)$ can be given by $$Z_{\mathcal{MD}}(G)=\frac{1}{2^g} \sum_{\beta\in \mathcal{S}} \sum_{q\in \mathcal{Q}(\overline{\Sigma})}\pm (-1)^{\text{Arf}(q)}\text{Pf}(M^{\psi(q)}_{\beta}(G)).$$ This version of our main result is interesting from geometrical point of view, however in the present work we will only focus on the more practical version stated in Theorem \ref{theo: main}, and will therefore not give further details.
  \end{remark}
  To conclude this section, we will show how the formula stated in Theorem \ref{theo: main} works in a concrete example.
  \begin{example}
      \begin{figure}[t]
\centering
  \includegraphics[height=130pt]{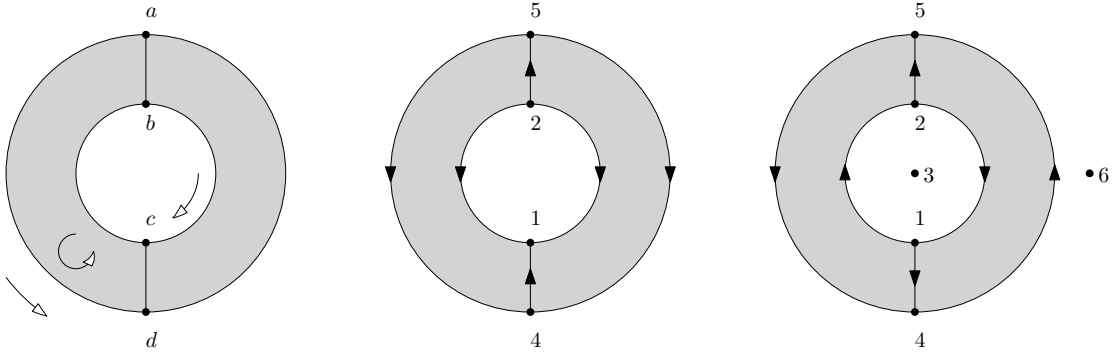}
  \caption{From left to right: the weighted graph $G$, the orientation $K_{0,0}$ and the orientation $K_{1,1}$ (together with a vertex labelling of $G$).} \label{fig: annulus}
\end{figure}
Let $G$ be the weighted graph embedded in the annulus (i.e. $g=0$ and $b=2$) illustrated in Figure \ref{fig: annulus}. Assume that the four vertices of $G$ are weighted $a,b,c,d$ while its edges are all weighted 1. Following the argument after Definition \ref{def: modified}, we label the vertices of $G$ corresponding to two parameters $(0,0)$ and $(1,1)$ and choose the specific orientations $K_{0,0}$ and $K_{1,1}$ as in Figure \ref{fig: annulus}. Then by Theorem \ref{theo: main} we have $$Z_{\mathcal{MD}}(G)=\big|\text{Pf}(M_{0,0}^{K_{0,0}})+ \text{Pf}(M_{1,1}^{K_{1,1}})\big|$$ where the two matrices $M_{0,0}^{K_{0,0}}$ and $M_{1,1}^{K_{1,1}}$ are given by $$M_{0,0}^{K_{0,0}}=\begin{pmatrix}
  0& -2-bc & -1&0\\
  2+bc& 0 & 0&1 \\
  1& 0& 0&-2-ad\\
  0& -1 & 2+ad&0
 \end{pmatrix}$$ and $$M_{1,1}^{K_{1,1}}=\begin{pmatrix}
  0& -bc & c&1&0&0\\
  bc& 0 & -b&0&1&0 \\
  -c& b& 0&0&0&0\\
  -1& 0 & 0&0&-ad&d\\
  0&-1&0&ad&0&-a\\
  0&0&0&-d&a&0
 \end{pmatrix}.$$ A direct calculation leads to $$\text{Pf}(M_{0,0}^{K_{0,0}})=5+2ad+2bc+abcd\qquad \text{Pf}(M_{1,1}^{K_{1,1}})=ab+cd$$ and so we obtain $$Z_{\mathcal{MD}}(G)=5+2ad+2bc+abcd+ab+cd.$$ As a reality check, one can verify easily that $\text{Pf}(M_{0,0}^{K_{0,0}})$ (resp. $\text{Pf}(M_{1,1}^{K_{1,1}})$) counts all the weighted MD coverings of $G$ with an even (resp. odd) number of monomers on each boundary component, and hence their sum is precisely equal to $Z_{\mathcal{MD}}(G)$.
  \end{example}
\section{Proof of the main result}\label{sec: proof} This section is devoted to prove Theorem \ref{theo: main} stated above. Since its proof needs some tools coming from the dimer model, let us first review some basis facts about this model and state some main results that we will need for our proofs. We refer readers to \cite{CimRes07} (resp. \cite{Cim09}) for the Pfaffian formula for the dimer partition function of graphs embedded in orientable (reps. possibly non-orientable) closed surfaces, and to \cite{CimRes08} for graphs embedded in surfaces with boundary.
\subsection{The dimer model on surface graphs} \label{subsec: dimer}
Let us first recall that the dimer partition function of the graph $G$ endowed with an edge weight system $x=(x_e)_{e\in E(G)}$ is given by
$$Z_{\mathcal{D}}(G)=\sum_{\tau_D\in \mathcal{D}(G)}\prod_{e\in \tau_D}x_e.$$
Recall also that an orientation $K$ on the edges of $G$ is called Kasteleyn if $n^K(\partial f)$ is odd for every face $f$ of $G$, where for $C$ an oriented cycle, $n^K(C)$ is the number of times that, travelling along $C$ by its orientation, we travel along an edge in the opposite direction to the one given by $K$. We will need the following result which can be found in \cite[ Proposition 1]{CimRes08}.
\begin{proposition} \label{pro: exist}
  Let $G$ be a graph embedded in an orientable surface $\Sigma$ of genus $g$ with $b$ boundary components $C_1,\dots,C_b$, and let $\beta_k$ be 0's or 1's for $1\leq k\leq b$. Then there exists a Kasteleyn orientation $K$ on $G$ such that $1+n^K(-C_k)\equiv \beta_k$ (modulo 2) for all $k$ if and only if $\sum_{k=1}^b \beta_k$ has the same parity as $|V(G)|$. \qed
\end{proposition}
Next we will recall the Pfaffian formula for the dimer partition function of surface graphs. Before doing that, let us quickly recall the definition of Pfaffians. Given a skew-symmetric matrix $A=(a_{ij})_{1\leq i,j\leq 2n}$, its $\emph{Pfaffian}$ is defined by $$\text{Pf}(A)=\frac{1}{2^n n!}\sum_{\sigma} \text{sign} (\sigma)a_{\sigma(1)\sigma(2)}\cdots a_{\sigma(2n-1)\sigma(2n)},$$ where the sum is over the set of all permutations $\sigma$'s of $\{1,\dots,2n\}$, and $\text{sign}(\sigma)=\pm1$ is the signature of $\sigma$. An important property of Pfaffians is that, if we add $\lambda$ times a row $r$ of $A$ to a row $s$, and do the same for corresponding columns, then the Pfaffian of $A$ does not change (cf. \cite[Corollary 4]{Austin09}). We will use this property later in our proof.

In this subsection, we are interested in the Pfaffian of adjacency matrices. These matrices can be defined as follows. If $(G,x)$ is an edge-weighted graph of $2n$ vertices labelled by $\{1,\dots, 2n\}$ and $K$ is an orientation on its edges, then the adjacency matrix of $G$ with respect to $K$, denoted by $A^K(G,x)=(a_{ij})_{1\leq i,j\leq 2n}$, has entries defined by $$a_{ij}=\sum_{e=(i,j)}\epsilon^K_{ij}(e)x(e).$$ Here the sum is taken over all the edge $e$ of $G$ between two vertices $i,j$, and $$\epsilon^K_{ij}(e)= \left\{  \begin{array}{ll}
        +1& \text{if $e$ is oriented by $K$ from $i$ to $j$};\\
        -1& \text{if not}.
\end{array} \right. $$

The Pfaffian formula for the dimer partition function of surface graphs can be stated as follows. Recall that the set $\{\alpha_i\}_{1\leq i\leq 2g}$ consists of $2g$ closed curves transverse to $G$ whose homology classes form a basis of $H_1(\overline{\Sigma};\mathbb{Z}_2)$, and $\gamma_i$ is an oriented cycle of $G$ having $\alpha_i$ on its immediate left for all $i$. Let $K$ be a Kasteleyn orientation on $G$ considered as a graph embedded in $\overline{\Sigma}$ such that $n^K(\gamma_i)$ is odd for every $1\leq i\leq 2g$. For any $\epsilon=(\epsilon_1,\dots,\epsilon_{2g})\in \mathbb{Z}_2^{2g}$, denote by $K^{\epsilon}$ the orientation obtained by inverting $K$ on every edge $e$ each time $e$ crosses $\alpha_i$ with $\epsilon_i=1$ for every $1\leq i\leq 2g$. With all these notations, we have the following result \cite[Theorem 3.9]{Cim09}.
\begin{theorem}\label{theo: dimer}
  Considering $(G,x)$ as an edge-weighted graph embedded in the closed surface $\overline{\Sigma}$, then the dimer partition function of $G$ is given by 
 $$ Z_{\mathcal{D}}(G)=\frac{1}{2^g}\bigg|\sum_{\epsilon\in \mathbb{Z}_2^{2g}}(-1)^{\sum_{i<j}\epsilon_i\epsilon_j\alpha_i\cdot\alpha_j} \text{Pf}(A^{K^{\epsilon}}(G,x))\bigg|,$$ where $A^{K^{\epsilon}}(G,x)$ is the adjacency matrix of $G$ with respect to the orientation $K^{\epsilon}$. \qed
\end{theorem}
We will use this theorem to prove our main result in the next subsection.
\subsection{Proof of Theorem \ref{theo: main}}\label{subsec: main proof} In this subsection we will use the preliminaries given in the previous part to prove Theorem \ref{theo: main}. The proof is based on Theorem \ref{theo: dimer} stated above together with some bijections between monomer-dimer coverings of $G$ and dimer coverings of some auxiliary graphs. We then prove that Pfaffians of adjacency matrices of these auxiliary graphs are equal to that of modified adjacency matrices of $G$ (recall Definition \ref{def: modified}) by some elementary matrix transformations together with the Laplace expansion for Pfaffians.

Before going into details, let us recall that $G$ is assumed to have the boundary circuits $B_k$ which coincide with the boundary components $C_k$ of the surface $\Sigma$. We also recall that the number $N_k$ of the vertices of $G$ on $B_k$ is even for every $1\leq k\leq b$, as well as the total number $|V(G)|$ of the vertices of $G$. We now begin by proving that $\mathcal{O}_{\beta}$ is non-empty, and that there exists $K_{\beta}$ as required (recall the arguments after Definition \ref{def: modified}).
\begin{lemma} \label{lem: exist}
  For each $\beta=(\beta_1,\dots,\beta_b)\in \mathbb{Z}_2^b$ with $\sum_{k=1}^b \beta_k$ even, the set $\mathcal{O}_{\beta}$ is nonempty. Moreover, there exists $K_{\beta}\in \mathcal{O}_{\beta}$ such that $n^{K_{\beta}}(\gamma_i)$ is odd for every $1\leq i\leq 2g$.
 \end{lemma}
 \begin{proof}
  Let us fix a $\beta=(\beta_1,\dots,\beta_b)\in \mathbb{Z}_2^b$ with $\sum_{k=1}^b \beta_k$ even. Note that the length of $B_k$ is equal to $N_k$ which is even, hence an element in $\mathcal{O}_{\beta}$ firstly must be some Kasteleyn orientation $L_{\beta}$ satisfying $n^{L_{\beta}}(-B_k)\equiv n^{L_{\beta}}(B_k)\equiv\beta_k+1$ modulo 2 for all $1\leq k\leq b$. Such an orientation $L_{\beta}$ clearly exists by Proposition \ref{pro: exist} since both $|V(G)|$ and $\sum_{k=1}^b\beta_k$ are even. Moreover, one can obtain an orientation belonging to $\mathcal{O}_{\beta}$ from such $L_{\beta}$ by the following transformation: one can transform $L_{\beta}$ to a Kasteleyn orientation $L_{\beta}'$ so that the restrictions of $L_{\beta}$ and $L_{\beta}'$ on $B_k$ are only different on any two arbitrary edges $e_1,e_2$ of $B_k$ ($1\leq k\leq b$). Then by repeating this transformation if needed, one obtains an element in $\mathcal{O}_{\beta}$. Now to construct $L_{\beta}'$ we can do as follows: draw a path $\gamma$ from the interior of the disk $D_k$ (whose boundary is $B_k$) first crossing transversely $e_1$, then some edges of $G$  and then $e_2$, finally coming back to the interior of $D_k$; invert $L_{\beta}$ on every edge each time this edge is crossed by $\gamma$ to obtain $L_{\beta}'$. It is easy to see that since $L_{\beta}$ is Kasteleyn, so is $L_{\beta}'$. Moreover by the construction $L_{\beta}$ and $L_{\beta}'$ restricted to $B_k$ are only different on $e_1,e_2$. We have proved the first part of the lemma.
  
 For the second part, to construct $K_{\beta}$ as required, let us pick an element $J_{\beta}\in \mathcal{O}_{\beta}$. If $n^{J_{\beta}}(\gamma_i)$ is odd for every $1\leq i\leq 2g$ then we are done. If there exists $i$ so that $n^{J_{\beta}}(\gamma_i)$ is even, let us pick a closed curve $\alpha_i^*$ transverse to $G$ and disjoint from $\partial \Sigma$ so that its homology class in $H_1(\overline{\Sigma};\mathbb{Z}_2)$ is dual to that of $\alpha_i$. Inverting $J_{\beta}$ on every edge $e$ each time $e$ is crossed by $\alpha_i^*$ results to change the parity of $n^{J_{\beta}}(\gamma_i)$ but not the parity of $n^{J_{\beta}}(\gamma_j)$ for every $j\neq i$. Repeating this procedure for each $i$ with $n^{J_{\beta}}(\gamma_i)$ even, we get the orientation $K_{\beta}$ as expected.
\end{proof}
\begin{figure}[b]
\centering
  \includegraphics[height=120pt]{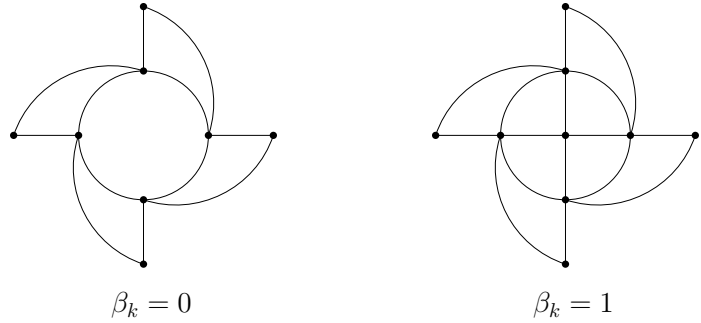}
  \caption{The shuriken graph $S_k$ corresponding to $\beta_k=0$ and $\beta_k=1$ for $N_k$=4.} \label{fig: shuriken}
\end{figure}
\begin{figure}[b]
\centering
  \includegraphics[height=120pt]{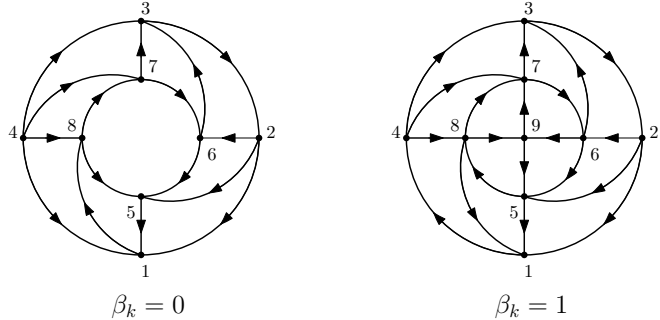}
  \caption{The vertex labelling and the orientation on $S_k$ added to the boundary circuit $B_k$ for $N_k=4$.} \label{fig: shuriken added}
\end{figure} 
We now continue by introducing some auxiliary graphs $G_{\beta}$ which depends on the parameter $\beta=(\beta_1,\dots,\beta_b)\in \mathbb{Z}_2^b$ with $\sum_{k=1}^b \beta_k$ even. Inspired by \cite[Appendix E]{Lieb16}, we define $G_{\beta}$ as the graph obtained from $G$ by adding a planar ``shuriken'' graph $S_k$ of $N_k$ ``blades'' on each boundary circuit $B_k$ of $G$ ($1\leq k\leq b$) so that all the ``blade-tip'' vertices of $S_k$ coincide with vertices on $B_k$  (see Figures \ref{fig: shuriken} and \ref{fig: shuriken added}). Note that there are two types of shuriken graphs, depending on whether $\beta_k$ is equal to 0 or 1. Moreover, if we label vertices on some boundary circuit $B_k$ of $G$ by $1,\dots,N_k$ then we would like to label the remaining vertices of $S_k$, called $\emph{inner vertices}$, by $N_k+1,\dots,2N_k$ (resp. by $N_k+1,\dots,2N_k+1$) if $\beta_k$ is even (resp. odd) consecutively following the inverse orientation of $B_k$. This labelling together with the specific labelling of the vertices of $G$ gives us a labelling of the vertices of $G_{\beta}$ that we will use from now on. Additionally, we can require that the vertex $j$ is adjacent to $j+N_k-1$ and $j+N_k$ for each $2\leq j\leq N_k$ and that the vertex 1 is adjacent to $N_k+1$ and $2N_k$ (see Figure \ref{fig: shuriken added}). The corresponding edges connecting these vertices will be called $\emph{blade edges}$, while the edges whose two endpoints are inner vertices will be called $\emph{inner edges}$. Furthermore we shall endow $G_{\beta}$ with an edge weight system as follows. Let $e$ be an edge of $G_{\beta}$. If $e$ is an edge of $G$ then it inherits the edge weight from $G$. If $e$ is a blade edge, then by definition $e$ is endowed with the vertex weight of its endpoint on the boundary circuit of $G$. Otherwise $e$ is weighted 1. From now on, without stated explicitly, we will always consider parameters $\beta=(\beta_1,\dots,\beta_b)\in \mathbb{Z}_2^b$ with $\sum_{k=1}^b \beta_k$ even so that $G_{\beta}$ has an even number of vertices. We also consider $G_{\beta}$ as a graph embedded in the closed surface $\overline{\Sigma}$. Note that we label the vertices of $G_{\beta}$ by the way mentioned above.

For $\beta=(\beta_1,\dots,\beta_b)\in \mathbb{Z}_2^b$, let us denote by $\mathcal{MD}^{\beta}(G)$ the set of all MD coverings of $G$ that consist of an even (resp. odd) number of monomers on $B_k$ with $\beta_k=0$ (resp. $\beta_k=1$) for every $1\leq k\leq b$. Setting the partial MD partition function of $G$ with respect to $\beta$ by $$Z_{\mathcal{MD}}^{\beta}(G):=Z_{\mathcal{MD}}^{\beta}(G,x,y):=\sum_{\tau\in \mathcal{MD}^{\beta}(G)}\prod_{e\in \tau_D}x_e\prod_{v\in \tau_M}y_v,$$ the following lemma indicates the relation between this partial partition function and the dimer partition function of $G_{\beta}$.
 \begin{lemma}\label{lem: bijection}
   For every $\beta=(\beta_1,\dots,\beta_b)\in \mathbb{Z}_2^b$ we have $$Z_{\mathcal{D}}(G_{\beta})=2^{\#\{k:\beta_k=0\}}\prod_{k:\beta_k=1}N_k\;Z_{\mathcal{MD}}^{\beta}(G).$$
 \end{lemma}
 \begin{proof}
 As mentioned before, we are interested in the case where $\sum_{k=1}^b \beta_k$ is even. Our aim is to show that each MD covering $\tau \in \mathcal{MD}^{\beta}(G)$ corresponds to $2^{\#\{k:\beta_k=0\}}\prod_{k:\beta_k=1}N_k$ dimer coverings $\tau_{\beta}$ of $G_{\beta}$ whose weights are equal to the weight of $\tau$. The idea is to keep the dimer part $\tau_{\mathcal{D}}$ of $\tau$ and match its monomer part $\tau_{\mathcal{M}}$ to vertices of $G_{\beta}\setminus G$ suitably. For the latter, it is clear that monomers of $\tau$ which lies on a boundary circuit $B_k$ of $G$ must be matched to inner vertices of the shuriken graph $S_k$. Therefore, we only need to prove that to each boundary circuit $B_k$ of $G$, there are precisely 2 possibilities to match monomers of $\tau$ if $\beta_k=0$, while there are precisely $N_k$ possibilities if $\beta_k=1$. In the former case, observe that every inner edge which is opposite to a monomer of $\tau$ can be ignored. Now if we remove from $S_k$ all such inner edges as well as all blade edges not adjacent to monomers of $\tau$, we are left with a circuit of even length, containing monomers of $\tau$. This circuit gives exactly 2 matchings covering these monomers and its remaining vertices. In the second case when $\beta_k=1$, we see that there are exactly $N_k$ possibilities to match the centre vertex of the shuriken graph $S_k$. Hence we only need to show that each such matching can be extended uniquely to a matching that covers monomers of $\tau$ lying on $B_k$ and remaining vertices of $S_k$. Similarly to the previous case, if we remove from $S_k$ its centre vertex and the inner vertex matched to it as well as all edges adjacent to these two, and also remove all inner edges that are opposite to monomers of $\tau$ together with  all blade edges not adjacent to monomers of $\tau$, we are left with a path of even length. This path, containing monomers of $\tau$, gives an unique matching covering them. Additionally, by definition the weight of $\tau$ and the weight of all the dimer coverings $\tau_{\beta}$'s constructed above are equal.

Finally to conclude our proof, one needs to show that the collection $\{\{\tau_{\beta}\}:{\tau \in \mathcal{MD}^{\beta}(G)}\}$ is a partition of $\mathcal{D}(G_{\beta})$. Indeed by the construction above if $\tau\neq \tau'$, then $\tau_{\mathcal{D}}\neq \tau'_{\mathcal{D}}$ and so $\{\tau_{\beta}\}$ and $\{\tau'_{\beta}\}$ are disjoint. Furthermore, if $D$ is a dimer covering of $G_{\beta}$, then $D$ covers an even (resp. odd) number of vertices of $B_k$ with $\beta_k=0$ (resp. $\beta_k=1$) for every $1\leq k\leq b$. Denote by $\tau$ the MD covering whose dimer part $\tau_{\mathcal{D}}$ coincides with $D$ restricted on $G$ and the monomer part $\tau_{\mathcal{M}}$ consists of all vertices on $\partial G$ which are matched by $D$ with some inner vertices of shuriken graphs. Then it is obvious that $\tau$ belongs to $\mathcal{MD}^{\beta}(G)$. Moreover, since $D$ matches the monomers of $\tau$ with some inner vertices of shuriken graphs, by the argument at the beginning of the proof, $D$ must coincide with one of the $\tau_{\beta}$'s.
\end{proof}
Next we will describe some specific orientations on $G_{\beta}$. For an orientation $L\in\mathcal{O}_{\beta}$, let us define $\overline{L}$ as the orientation on $G_{\beta}$ which coincides with $L$ on $G$ and is determined on the remaining edges of $G_{\beta}$ as follows. Recalling the labelling of vertices of $G_{\beta}$ mentioned before, if $\beta_k=1$, on blade edges of $S_k$ we require $\overline{L}$ to go out from $j$ if $j$ is even, and go toward $j$ if $j$ is odd for all $1\leq j\leq N_k$. In addition, on inner edges, we require $\overline{L}$ to be from big-labelled vertices to small-labelled ones except that it goes from $N_k+1$ to $2N_k$, whilst it goes from the centre vertex of $S_k$ to odd-labelled inner vertices, and toward this centre vertex from even-labelled ones. If $\beta_k=0$, we require $\overline{L}$ on blade edges to be the same as in the previous case, except that we invert the orientation on the edge between 1 and $2N_k$. On inner edges, we would like $\overline{L}$ to be from big-labelled vertices to small-labelled ones. This is illustrated in Figure \ref{fig: shuriken added}. It is straightforward to verify that the orientation $\overline{L}$ defined by this way is Kasteleyn on $G_{\beta}\subset \overline{\Sigma}$ for every $L\in\mathcal{O}_{\beta}$.

Recall that we label vertices of $G$ in a specific way (cf. the argument after Definition \ref{def: modified}), that induces a vertex labelling of $G_{\beta}$ (described after Lemma \ref{lem: exist}). With these vertex labels we have the following result.
\begin{proposition} \label{pro: matrix}
  For $L\in\mathcal{O}_{\beta}$, let $A^{\overline{L}}(G_{\beta})$ be the adjacency matrix of $G_{\beta}$ with respect to the orientation $\overline{L}$. We have $$\text{Pf}(A^{\overline{L}}(G_{\beta}))=(-1)^{\sum_{k=1}^b N_k/2}2^{\#\{k:\beta_k=0\}}\prod_{k:\beta_k=1}N_k\;\text{Pf}(M_{\beta}^{L}(G)).$$
 \end{proposition}
 The proof of this proposition will be left until the end of this subsection. We now show how this proposition implies Theorem \ref{theo: main}.
 \begin{proof}[Proof of Theorem \ref{theo: main}] Recall that $\{\alpha_i\}_{1\leq i\leq 2g}$ is a set of closed curves whose homology classes form a basis of $H_1(\overline{\Sigma};\mathbb{Z}_2)$, and that $\gamma_i$ is an oriented closed cycle of $G$ which is disjoint from $\partial G$ and has $\alpha_i$ on its immediate left for every $1\leq i\leq 2g$. Recall also that $K_{\beta}$ is an element of $\mathcal{O}_{\beta}$ satisfying that $n^{K_{\beta}}(\gamma_i)$ is odd for every $i$. 
   By definition, $\overline{K_{\beta}}$ is Kasteleyn on $G_{\beta}$ and still satisfies that $n^{\overline{K_{\beta}}}(C_i)$ is odd.
 Moreover, since $\alpha_i$ is chosen to be disjoint from $\partial G$, it is clear that $\overline{K_{\beta}}^{\epsilon}=\overline{K_{\beta}^{\epsilon}}$. Using this fact and applying Theorem \ref{theo: dimer} for $G_{\beta}$ and $\overline{K_{\beta}}$ we get 
 \begin{eqnarray}
   Z_{\mathcal{D}}(G_{\beta})&=&\frac{1}{2^g}\bigg|\sum_{\epsilon\in \mathbb{Z}_2^{2g}}(-1)^{\sum_{i<j}\epsilon_i\epsilon_j\alpha_i\cdot\alpha_j} \text{Pf}(A^{\overline{K_{\beta}}^{\epsilon}}(G_{\beta}))\bigg| \nonumber\\
   &=&\frac{1}{2^g}\bigg|\sum_{\epsilon\in \mathbb{Z}_2^{2g}}(-1)^{\sum_{i<j}\epsilon_i\epsilon_j\alpha_i\cdot\alpha_j} \text{Pf}(A^{\overline{K_{\beta}^{\epsilon}}}(G_{\beta}))\bigg|.\nonumber
   \end{eqnarray}  
    In the following, the sum over $\beta$ will be understood as over $\beta=(\beta_1,\dots,\beta_b)\in \mathbb{Z}_2^b$ such that $\sum_{k=1}^b\beta_k$ is even. With this convention and by the equality above we can write 
  \begin{eqnarray}
     Z_{\mathcal{MD}}(G)&=&\sum_{\beta}Z_{\mathcal{MD}}^{\beta}(G)\nonumber\\&\overset{Lem. \ref{lem: bijection}}{=}&\sum_{\beta}\bigg(2^{\#\{k:\beta_k=0\}}\prod_{k:\beta_k=1}N_k\bigg)^{-1}Z_{\mathcal{D}}(G_{\beta})\nonumber\\
     &=&\sum_{\beta}\bigg(2^{\#\{k:\beta_k=0\}}\prod_{k:\beta_k=1}N_k\bigg)^{-1}\bigg|\sum_{\epsilon\in \mathbb{Z}_2^{2g}}(-1)^{\sum_{i<j}\epsilon_i\epsilon_j\alpha_i\cdot\alpha_j} \text{Pf}(A^{\overline{K^{\epsilon}_{\beta}}}(G_{\beta}))\bigg| \nonumber\\
     &\overset{Prop. \ref{pro: matrix}}{=}&\frac{1}{2^g} \sum_{\beta}\bigg| \sum_{\epsilon\in\mathbb{Z}_2^{2g}}(-1)^{\sum_{i<j}\epsilon_i\epsilon_j\alpha_i\cdot\alpha_j}\text{Pf}(M^{K_{\beta}^{\epsilon}}_{\beta}(G))\bigg|. \nonumber
   \end{eqnarray}
   This concludes the proof of Theorem \ref{theo: main}.
   \end{proof}
 Now we are only left with the proof of Proposition \ref{pro: matrix}. The idea of the proof is that, for each $K\in \mathcal{O}_{\beta}$, we can transform the matrix $A^{\overline{K}}(G_{\beta})$ to a new matrix using only elementary row-column operations so that their Pfaffians are equal, and then we relate the Pfaffian of the latter to that of $M^K_{\beta}(G)$ using Laplace expansions. Hence before giving the proof of Proposition \ref{pro: matrix}, let us recall the Laplace expansion for Pfaffians, whose proof can be found in \cite[Proposition 2.3]{Ishi06}.
 \begin{lemma}\label{lem: Laplace}
   If $A=(a_{ij})$ is a skew-symmetric matrix of size $2n$, then for any $i=1,\dots,2n$ we have $$\text{Pf}(A)=\sum_{\substack{j=1\\j\neq i}}^{2n}(-1)^{i+j+1+\theta(i-j)}a_{ij}\text{Pf}(A_{\hat{i}\hat{j}}),$$ where $A_{\hat{i}\hat{j}}$ is the matrix obtained from $A$ by removing both $i^{\text{th}}$ and $j^{\text{th}}$ rows and columns, while $\theta$ is the $\emph{Heaviside step function}$, that is, $\theta(l)$ is equal to 1 if $l\geq 0$ and 0 otherwise.
  \end{lemma}
 \begin{proof}[Proof of Proposition \ref{pro: matrix}]
   As mentioned above, for an orientation $K\in \mathcal{O}_{\beta}$, we will first transform the matrix $A^{\overline{K}}(G_{\beta})$ to a new matrix using row-column operations, and then use Laplace expansions to relate the Pfaffian of the latter to that of $M^K_{\beta}(G)$. The point is that all these operations and expansions, as we will see, are local and only depend on boundary circuits individually. Therefore, without loss of generality, in the following we will only work with the matrices $A^{\overline{K}}(G_{\beta})$ and $M^K_{\beta}(G)$ partially, that is, we will work with their submatrices corresponding to each of boundary circuits. However readers should keep in mind that what we will show now works completely well in the global context. Throughout the rest of this proof, let us fix a $\beta\in \mathbb{Z}_2^b$ with $\sum_{k=1}^b\beta_k$ even, and an orientation $K\in \mathcal{O}_{\beta}$.
   
Let us begin with a boundary circuit $B_k$ corresponding to $\beta_k=0$. Recalling the labelling of $G_{\beta}$ on this boundary circuit as well as the orientation $\overline{K}$, we can write the submatrix corresponding to $B_k$ of $A^{\overline{K}}(G_{\beta})$ as a block matrix $$A_k=\left(\begin{array}{cc} A^{\partial} & A\\ -A^T&A^{\text{in}}
 \end{array}\right).$$ Here $A^{\partial}$ represents the adjacencies between vertices on $B_k$, $A^{\text{in}}$ represents the adjacencies between inner vertices of the shuriken graph $S_k$, while $A$ represents adjacencies between vertices of these two types. Note that all these matrices are of size $N_k$. More precisely, denoting vertex weights of vertices on $B_k$ by $y_1,\dots,y_{N_k}$ we can write $A=(a_{ij})_{1\leq i,j\leq N_k}$ with all of entries equal to 0 except $a_{11}=-y_1=-a_{1,N_k}$ and $a_{i,i-1}=a_{ii}=(-1)^{i}y_i$ for $2\leq i\leq N_k$. Also we have $A^{\text{in}}=(b_{ij})_{1\leq i,j\leq N_k}$ with $b_{1,N_k}=-b_{N_k,1}=-1$, $b_{i,i+1}=-b_{i+1,i}=-1$ for $1\leq i\leq N_k-1$ and other entries equal to 0. Denoting by $M_k$ the submatrix of $M^K_{\beta}(G)$ corresponding to $B_k$ and writing $N_k=2n_k$ for further purposes, our aim now is to prove that 
 \begin{equation}\label{eq: 0}
   \text{Pf}(A_k)=2(-1)^{n_k}\text{Pf}(M_k).
 \end{equation}
To do so, we will first transform the matrix $A_k$ using row and column operations so that its Pfaffian does not change. Let us denote by $R_{l}(X)$ and $C_{l}(X)$ the $l^{\text{th}}$ row and the $l^{\text{th}}$ column of a matrix $X$.  Add $(-1)^ly_lR_{N_k+1}(A_k)$ to $R_l(A_k)$ as well as $(-1)^ly_lC_{N_k+1}(A_k)$ to $C_l(A_k)$ for each $2\leq l\leq N_k$ step by step. Since $b_{ij}=0$ if both $i,j$ are odd, one can verify easily that after each step all the odd columns of $A$ (as well as odd rows of $-A^T$) do not change. By the same reason, for each $1\leq m\leq n_k-1$ we can add $(-1)^{l}y_lR_{N_k+2m+1}(A_k)$ to $R_l(A_k)$ (and similarly for corresponding columns) for every $2m+2\leq l\leq N_k$. Note that these operations keep $A^{\text{in}}$ unchanged. Moreover, by the definition of $M^K_{\beta}(G)$, it is straightforward to check that after all these operations above, the matrix $A_k$ becomes $$A'_k=\left(\begin{array}{cc} M_k & A'\\ -A'^T&A^{\text{in}}
 \end{array}\right)$$ whose Pfaffian is equal to the Pfaffian of $A_k$. Let us now determine the entries $a'_{ij}$ of $A'$. By the argument above, all the odd columns of $A'$ are equal to those of $A$. We will show that all the even columns of $A'$ are 0 except the last one. Indeed, for $1\leq j\leq n_k-1$ we have $$a'_{2j,2j}=a_{2j,2j}+y_{2j}(b_{1,2j}+b_{3,2j}+\cdots+b_{2j-3,2j}+b_{2j-1,2j})=0$$ since $a_{2j,2j}=y_{2j}$, $b_{2j-1,2j}=-1$ while $b_{1,2j}=\cdots=b_{2j-1,2j}=0$. Similarly we have $$a'_{2j+1,2j}=a_{2j+1,2j}-y_{2j+1}(b_{1,2j}+b_{3,2j}+\cdots+b_{2j-3,2j}+b_{2j-1,2j})=0$$ as $a_{2j+1,2j}=-y_{2j+1}$. If $l<2j$ then we have $$a'_{l,2j}=a_{l,2j}+(-1)^ly_{l}(b_{1,2j}+b_{3,2j}+\cdots+b_{p,2j})$$ with some $p$ odd and $p\leq 2j-3$. Since in this case $a_{l,2j}=0$, we get $a'_{l,2j}=0$. If $l>2j+1$, we can write $$a'_{l,2j}=a_{l,2j}+(-1)^ly_{l}(b_{1,2j}+\cdots+b_{2j-1,2j}+b_{2j+1,2j}+\cdots+b_{q,2j})$$ with some $q$ odd and $q\geq 2j+3$. As in this case we also have $a_{l,2j}=0$, $b_{2j-1,2j}=-b_{2j+1,2j}=-1$ while others $b_{i,2j}$'s are 0, we obtain $a'_{l,2j}=0$ as well. Finally let us look at the last column of $A'$. Since all the entries of this column are determined only by the last entries of $R_{N_k+1}(A_k)$ and $R_{2N_k-1}(A_k)$ together with the last column of $A$, we simply get $C_{N_k}(A')=(y_1,-y_2,\dots,y_{N_k-1},-y_{N_k})^T$.\\
 Now let us make one more transformation. We add $\sum_{m=1}^{n_k}R_{N_k+2m-1}(A_k')$ to $R_{2N_k}(A'_k)$ as well as $\sum_{m=1}^{n_k}C_{N_k+2m-1}(A_k')$ to $C_{2N_k}(A'_k)$ so that we get $$A''_k=\left(\begin{array}{cc} M_k & A''\\ -A''^T&A^{\text{in}}
 \end{array}\right).$$ Note that now the block $A''$ of $A''_k$ has all even columns equal to 0, and we still have $\text{Pf}(A''_k)=\text{Pf}(A'_k)=\text{Pf}(A_k)$. Applying Lemma \ref{lem: Laplace} for the last row of $A''_k$ in which there are only 2 entries different from 0, we get $$\text{Pf}(A''_k)=-2\text{Pf}\left(\begin{array}{cc} M_k & D\\ -D^T& E \end{array}\right).$$ Here $E$ is the matrix obtained from $A_{\text{in}}$ by removing its last two rows and last two columns, while $D$ is a matrix whose even columns consist of 0's. Repeating this expansion for the new block matrix with a remark that now its last row contains only one nonzero element, and using the fact that the top right block has all even columns equal to 0, by recursion we get $$\text{Pf}(A_k)=\text{Pf}(A''_k)=2(-1)^{n_k}\text{Pf}(M_k).$$ This concludes the proof of Equation (\ref{eq: 0}).
 
Next we consider a boundary circuit $B_k$ corresponding to $\beta_k=1$. In this case the corresponding submatrix of $A^{\overline{K}}(G_{\beta})$ to this circuit is given by $$A_k^1=\left(\begin{array}{cc} A^{\partial}_1 & A_1\\ -(A_1)^T&A^{\text{in}}_1
 \end{array}\right)$$ Due to the differences of $S_k$ and $\overline{K}$ in this case, the block $A_1$ is obtained from $A$ as follows: we change the sign of $a_{1,N_k}$, and then add a column of all 0's to the right of $A$. Similarly, the matrix $A_1^{\text{in}}$ is obtained from $A^{\text{in}}$ by changing the signs of $b_{1,N_k}$ and $b_{N_k,1}$, and by adding the column $(-1,1,\dots,-1,1,0)^T$ to its rights and the row $(1,-1,\dots,1,-1,0)$ to its bottom. Our purpose now is to prove that \begin{equation}\label{eq: 1}
 \text{Pf}(A^1_k)=(-1)^{n_k}N_k\text{Pf}(M_k^1).
 \end{equation}
For each $0\leq m\leq n_k-1$ let us add $(-1)^{l}y_lR_{N_k+2m+1}(A^1_k)$ to $R_l(A^1_k)$ and $(-1)^{l}y_lC_{N_k+2m+1}(A^1_k)$ to $C_l(A^1_k)$ for every $2m+2\leq l\leq N_k$ as in the previous case. By the same argument we obtain 
$$(A_k^1)'=\left(\begin{array}{cc} M_k & A_1'\\ -(A_1')^T&A_1^{\text{in}}
 \end{array}\right)$$ 
 so that $\text{Pf}(A_k^1)'=\text{Pf}(A_k^1)$. However note that here $M_k$ is only the submatrix of $M_{\beta}^K(G)$ corresponding to the circuit $B_k$ without adding the vertex $v_k$. More precisely, $M_k^1$ is obtained from $M_k$ by adding the column $(y_1,-y_2,\dots,y_{N_k-1},-y_{N_k})^T$ to its right and the row $(-y_1,y_2,\dots,-y_{N_k-1},y_{N_k})$ to its bottom. Also $A_1'$ and $A'$ coincide on the first $N_k-1$ columns, while their $N_k^{\text{th}}$ columns are of opposite signs. Moreover, following our operations above and by some simple calculations, we can find that the last column of $A_1'$ is $$(0, -y_2,y_3,\dots, -iy_{2i},iy_{2i+1},\dots, -(n_k-1)y_{N_k-2}, (n_k-1)y_{N_k-1},-n_ky_{N_k})^T.$$ Now let us do some more transformations. First of all, similarly to the previous case we add $-\sum_{m=1}^{n_k}R_{N_k+2m-1}(A^1_k)'$ to $R_{N_k+2n_k}(A_k^1)'$ (and do similarly for corresponding columns) to get $$(A_k^1)''=\left(\begin{array}{cc} M_k & A_1''\\ -(A_1'')^T&(A_1^{\text{in}})'
 \end{array}\right)$$ so that $\text{Pf}(A_k^1)''=\text{Pf}(A_k^1)'$ and $A_1''$ has all even columns equal to 0. Also, $(A_1^{\text{in}})'$ can be obtained from $A^{\text{in}}$ by adding the column $(-1,1,\dots,-1,n_k+1,0)^T$ to its right and the row $(1,-1,\dots,1,-n_k-1,0)$ to its bottom. Secondly we add $\sum_{m=1}^{n_k}mR_{N_k+2m-1}(A^1_k)''$ to $R_{2N_k+1}(A^1_k)''$ (and do similarly for corresponding columns) to obtain $$(A_k^1)'''=\left(\begin{array}{cc} M_k & A_1'''\\ -(A_1''')^T&(A_1^{\text{in}})''
 \end{array}\right)$$ so that $\text{Pf}(A_k^1)'''=\text{Pf}(A_k^1)''$, $A_1'''$ still has all even columns equal to 0 while $(A_1^{\text{in}})''$ is obtained  by adding the columns $(-1,0,-1,0,\dots,-1,N_k,0)^T$ to the right of $A^{\text{in}}$ and adding the row $(1,0,1,0,\dots,1,-N_k,0)$ to its bottom. Finally let us add $\sum_{m=2}^{n_k}R_{N_k+2m-1}(A_k^1)'''$ to $R_{N_k+1}(A_k^1)'''$ (and similarly for columns) to get $$A_{\text{final}}=\left(\begin{array}{cc} M_k & P\\ -P^T&Q
 \end{array}\right)$$ so that $\text{Pf}(A_{\text{final}})=\text{Pf}(A_k^1)'''$ and the matrix $P$ has all even columns equal to 0, while its first column is $(-y_1,y_2,\dots,-y_{N_k-1},y_{N_k})^T$. Also one gets $$Q = 
 \begin{pmatrix}
  0& 0 & 0&0& \cdots &0&0&* \\
  0& 0 & -1&0& \cdots &0&0&0 \\
  0& 1 & 0&-1& \cdots &0&0&* \\
  0& 0 & 1&0& \cdots &0&0&0 \\
  \vdots  & \vdots  &\vdots &\vdots& \ddots & \vdots&\vdots&\vdots  \\
  0 &0 &0&0& \cdots & 0&-1&*\\
   0 &0 &0&0& \cdots & 1&0&N_k\\
   * &0 &*&0& \cdots & *&-N_k&0
 \end{pmatrix}$$ is a skew-symmetric matrix of size $N_k+1$ (where $*$ indicates non-important entries). Now applying Lemma \ref{lem: Laplace} for the $(2N_k)^{\text{th}}$ row of $A_{\text{final}}$ in which there are exactly two entries different from 0, namely 1 and $N_k$, we get \begin{equation}\label{eq: final}
 \text{Pf}(A_{\text{final}})=N_k\text{Pf} \begin{pmatrix} M_k & U\\ -U^T& X_{N_k-1}
 \end{pmatrix}-\text{Pf} \begin{pmatrix} M_k & V\\ -V^T& Y_{N_k-1}
 \end{pmatrix}.
 \end{equation}
  In this equation, $U,V$ are matrices with all even columns equal to 0 while $X_m$ is the matrix of size $m$ of the following type  $$\begin{pmatrix}
  0& 0 & 0&0& \cdots &0&0\\
  0& 0 & -1&0& \cdots &0&0 \\
  0& 1 & 0&-1& \cdots &0&0 \\
  0& 0 & 1&0& \cdots &0&0 \\
  \vdots  & \vdots  &\vdots &\vdots& \ddots & \vdots& \vdots  \\
 0 &0&0&0& \cdots &0&-1\\
 0 &0&0&0& \cdots &1&0
 \end{pmatrix}.$$ Additionally $Y_m$ is obtained by adding the column of type $(*,0,\dots,*,0,0)^T$ to the right of $X_{m-1}$ and the row of same type to its bottom. To compute the first Pfaffian on the right hand side of (\ref{eq: final}), one can apply Lemma \ref{lem: Laplace} for the row corresponding to the second row of $X_{N_k-1}$ in which the only nonzero element is -1. Note that by removing rows and columns corresponding to this entry, we obtain a new matrix of the same type as before with size decreased by 2. By induction, and by changing the sign of the last column (and the last row) of the final matrix, we get the first term on the right hand side of (\ref{eq: final}) exactly equal to $N_k(-1)^{n_k}\text{Pf}(M_k^1)$. By the same argument one can compute the second term of (\ref{eq: final}) equal to 0. This leads to $$\text{Pf}(A^1_k)=\text{Pf}(A^1_k)'=\text{Pf}(A^1_k)''=\text{Pf}(A^1_k)'''=\text{Pf}(A_{\text{final}})=N_k(-1)^{n_k}\text{Pf}(M^1_k)$$ which proves Equation (\ref{eq: 1}).
 
Finally, to conclude our proof, one only needs to combine Equation (\ref{eq: 0}) and (\ref{eq: 1}) together, and remark that all the operations and expansions we show above work completely well in the global context. This concludes the proof of Proposition \ref{pro: matrix}, as well as of Theorem \ref{theo: main}.
\end{proof}
 \bibliographystyle{plain}
\bibliography{reference}

\begin{thebibliography}{10}

\bibitem{Austin09}
Tracale Austin, Hans Bantilan, Eric~S Egge, Isao Jonas, and Paul Kory.
\newblock The pfaffian transform.
\newblock {\em Journal of Integer Sequences}, 12(2):3, 2009.

\bibitem{Cim09}
David Cimasoni.
\newblock Dimers on graphs in non-orientable surfaces.
\newblock {\em Lett. Math. Phys.}, 87(1-2):149--179, 2009.

\bibitem{CimRes07}
David Cimasoni and Nicolai Reshetikhin.
\newblock Dimers on surface graphs and spin structures. {I}.
\newblock {\em Comm. Math. Phys.}, 275(1):187--208, 2007.

\bibitem{CimRes08}
David Cimasoni and Nicolai Reshetikhin.
\newblock Dimers on surface graphs and spin structures. {II}.
\newblock {\em Comm. Math. Phys.}, 281(2):445--468, 2008.

\bibitem{Fis61}
Michael~E. Fisher.
\newblock Statistical mechanics of dimers on a plane lattice.
\newblock {\em Phys. Rev.}, 124:1664--1672, Dec 1961.

\bibitem{Fowler37}
RH~Fowler and GS~Rushbrooke.
\newblock An attempt to extend the statistical theory of perfect solutions.
\newblock {\em Transactions of the Faraday Society}, 33:1272--1294, 1937.

\bibitem{Gal99}
Anna Galluccio and Martin Loebl.
\newblock On the theory of pfaffian orientations. {I}. perfect matchings and
  permanents.
\newblock {\em Electron. J. combin}, 6(1):R6, 1999.

\bibitem{Garey02}
Michael~R Garey and David~S Johnson.
\newblock {\em Computers and intractability}, volume~29.
\newblock wh freeman New York, 2002.

\bibitem{Lieb16}
Alessandro Giuliani, Ian Jauslin, and Elliott~H Lieb.
\newblock A pfaffian formula for monomer--dimer partition functions.
\newblock {\em Journal of Statistical Physics}, 163(2):211--238, 2016.

\bibitem{Hielmann-Lieb70}
Ole~J Heilmann and Elliott~H Lieb.
\newblock Monomers and dimers.
\newblock In {\em Statistical Mechanics}, pages 41--43. Springer, 1970.

\bibitem{Heilmann-Lieb72}
Ole~J Heilmann and Elliott~H Lieb.
\newblock Theory of monomer-dimer systems.
\newblock In {\em Statistical Mechanics}, pages 45--87. Springer, 1972.

\bibitem{Ishi06}
Masao Ishikawa and Masato Wakayama.
\newblock Applications of minor summation formula {III}, {P}l{\"u}cker
  relations, lattice paths and {P}faffian identities.
\newblock {\em Journal of Combinatorial Theory, Series A}, 113(1):113--155,
  2006.

\bibitem{Jer87}
Mark Jerrum.
\newblock Two-dimensional monomer-dimer systems are computationally
  intractable.
\newblock {\em Journal of Statistical Physics}, 48(1):121--134, 1987.

\bibitem{Kas61}
Pieter~W Kasteleyn.
\newblock The statistics of dimers on a lattice: {I}. the number of dimer
  arrangements on a quadratic lattice.
\newblock {\em Physica}, 27(12):1209--1225, 1961.

\bibitem{Kas63}
Pieter~W Kasteleyn.
\newblock Dimer statistics and phase transitions.
\newblock {\em Journal of Mathematical Physics}, 4(2):287--293, 1963.

\bibitem{Priezzhev08}
Vyatcheslav~B Priezzhev and Philippe Ruelle.
\newblock Boundary monomers in the dimer model.
\newblock {\em Physical Review E}, 77(6):061126, 2008.

\bibitem{Tem73}
Harold~NV Temperley.
\newblock Enumeration of graphs on a large periodic lattice.
\newblock In {\em Combinatorics: Proceedings of the British Combinatorial
  Conference}, pages 155--159, 1973.

\bibitem{Tes2000}
Glenn Tesler.
\newblock Matchings in graphs on non-orientable surfaces.
\newblock {\em Journal of Combinatorial Theory, Series B}, 78(2):198--231,
  2000.

\bibitem{Wu06}
FY~Wu.
\newblock Pfaffian solution of a dimer-monomer problem: single monomer on the
  boundary.
\newblock {\em Physical Review E}, 74(2):020104, 2006.

\bibitem{WuTzeng11}
FY~Wu, Wen-Jer Tzeng, and N~Sh Izmailian.
\newblock Exact solution of a monomer-dimer problem: A single boundary monomer
  on a nonbipartite lattice.
\newblock {\em Physical Review E}, 83(1):011106, 2011.

\end{thebibliography}
 \end{document}